\pdfoutput=1

\documentclass[letterpaper, 10pt, journal]{IEEEtran}

\IEEEoverridecommandlockouts       

\usepackage{amsmath,amsfonts}
\usepackage{algorithm}
\usepackage{array}
\usepackage[caption=false,font=normalsize,labelfont=sf,textfont=sf]{subfig}
\usepackage{textcomp}
\usepackage{stfloats}
\usepackage{url}
\usepackage{verbatim}
\usepackage{graphicx}
\usepackage{cite}
\usepackage{xcolor}
\usepackage{booktabs}
\usepackage{algpseudocode}

\usepackage{amsthm}
\usepackage{hyperref}
\usepackage{amssymb}
\usepackage{yfonts}
\usepackage{cleveref}

\usepackage[T1]{fontenc}
\usepackage{mathpazo}
\usepackage{microtype}

\usepackage[letterpaper, top=1in, bottom=1in, left=0.875in, right=0.875in]{geometry}
\usepackage{etoolbox}
\patchcmd{\abstract}
  {\bfseries\textit{\abstractname}---\relax}
  {\textbf{\textit{\abstractname}}---\relax}
  {}{\typeout{(preprint restyling) abstract patch skipped}}

\providecommand{\IEEEPARstart}[2]{\PARstart{#1}{#2}}

\makeatletter
\@ifundefined{IEEEkeywords}{%
    {\vspace{0.4\baselineskip}\noindent\small\itshape\textbf{\textit{Index Terms}}---}%
    {\par\normalsize}%
}{}
\makeatother

\hypersetup{pdfauthor={},pdfsubject={},pdfkeywords={}}

\definecolor{prsapphire}{HTML}{1F4E8C}
\definecolor{prpine}{HTML}{0B6E4F}
\definecolor{prburgundy}{HTML}{8C2F39}
\hypersetup{colorlinks=true, linkcolor=prsapphire, citecolor=prpine, urlcolor=prburgundy}

\definecolor{prink}{HTML}{22282F}      
\definecolor{prcloud}{HTML}{F4F5F7}    
\definecolor{prcomment}{HTML}{9AA3AD}  
\definecolor{prlineno}{HTML}{AEB6BF}   

\usepackage[most]{tcolorbox}

\newtheoremstyle{prsober}{2pt}{2pt}{\itshape}{}{\bfseries}{.}{ }{\thmname{#1}\thmnumber{ #2}\thmnote{ {\mdseries\itshape(#3)}}}
\theoremstyle{prsober}

\newtheorem{theorem}{Theorem}
\newtheorem{proposition}{Proposition}
\newtheorem{corollary}{Corollary}
\newtheorem{assumption}{Assumption}

\tcbset{prstatement/.style={enhanced jigsaw, breakable, frame hidden,
  boxrule=0pt, sharp corners, colback=prcloud,
  borderline west={1.4pt}{0pt}{prink},
  left=7pt, right=6pt, top=4pt, bottom=4pt,
  before skip=8pt plus 2pt, after skip=8pt plus 2pt}}
\tcolorboxenvironment{theorem}{prstatement}
\tcolorboxenvironment{corollary}{prstatement}
\tcolorboxenvironment{proposition}{prstatement}
\tcolorboxenvironment{lemma}{prstatement}
\tcolorboxenvironment{assumption}{prstatement}
\tcolorboxenvironment{remark}{enhanced jigsaw, breakable, frame hidden,
  boxrule=0pt, sharp corners, colback=prcloud,
  left=7pt, right=6pt, top=4pt, bottom=4pt,
  before skip=8pt plus 2pt, after skip=8pt plus 2pt}

\tcolorboxenvironment{abstract}{enhanced jigsaw, frame hidden, boxrule=0pt,
  sharp corners, colback=prcloud, borderline west={1.4pt}{0pt}{prink},
  left=7pt, right=6pt, top=5pt, bottom=5pt, before skip=6pt, after skip=8pt}

\makeatletter
\long\def\@makecaption#1#2{%
  \ifx\@captype\@IEEEtablestring%
    \footnotesize\bgroup\par\centering
    {\normalfont\footnotesize\scshape #1}\\[0.25ex]
    {\normalfont\footnotesize #2}\par\egroup
    \vskip\belowcaptionskip
  \else
    \vskip\abovecaptionskip
    \setbox\@tempboxa\hbox{\normalfont\footnotesize{\bfseries#1.}\enspace #2}%
    \ifdim \wd\@tempboxa >\hsize
      {\normalfont\footnotesize{\bfseries#1.}\enspace #2\par}%
    \else
      \hbox to\hsize{\hfil\box\@tempboxa\hfil}%
    \fi
  \fi}
\makeatother

\makeatletter
\newcommand\fs@prruled{\def\@fs@cfont{\bfseries}\let\@fs@capt\floatc@ruled
  \def\@fs@pre{{\color{prink}\hrule height0.8pt depth0pt}\kern3pt}%
  \def\@fs@post{\kern2.5pt{\color{prink}\hrule height0.8pt}\relax}%
  \def\@fs@mid{\kern2.5pt{\color{prink!35}\hrule height0.5pt}\kern3pt}%
  \let\@fs@iftopcapt\iftrue}
\makeatother
\floatstyle{prruled}
\restylefloat{algorithm}
\tcolorboxenvironment{algorithmic}{enhanced, frame hidden, boxrule=0pt,
  sharp corners, colback=prcloud, left=2pt, right=2pt, top=4pt, bottom=4pt,
  before skip=2pt, after skip=2pt}
\algrenewcommand\alglinenumber[1]{{\footnotesize\color{prlineno}#1:}}
\algrenewcommand\algorithmiccomment[1]{\hfill{\footnotesize\itshape\color{prcomment}\(\triangleright\) #1}}

\newcommand{\parhead}[1]{\textbf{#1}}

\newcommand{\diff}{\mathop{}\!{\mathrm{d}}}

\newcommand{\T}{{\top}}
\newcommand{\R}{\mathbb{R}}
\newcommand{\E}{\mathbb{E}}
\newcommand{\N}{\mathbb{N}}

\DeclareMathOperator*{\argmin}{arg\,min}

\begin{document}

\title{\LARGE \bf
Trajectory Optimization via Schr\"{o}dinger Bridge Sampling%
}

\author{Mattia Mosso$^{1,*}$ ~~ Yang Liu$^{2,*}$ ~~ Heng Yang$^{3}$
\thanks{$^*$ Equal contribution.}
\thanks{$^{1}$ School of Aerospace Engineering, Georgia Institute of Technology. Email: \texttt{mmosso3@gatech.edu}}
\thanks{$^{2}$ College of Science and Engineering, University of Minnesota. Email: \texttt{liu03222@umn.edu}}
\thanks{$^{3}$ School of Engineering and Applied Sciences, Harvard University. Email: \texttt{hankyang@seas.harvard.edu}}%
}
\maketitle

\begin{abstract}
We take a new look at the relation between finite-horizon trajectory optimization and Schr\"odinger bridge sampling. Viewed as inference, KL-regularized trajectory optimization is solved by sampling from a Gibbs--Boltzmann distribution whose energy is the trajectory cost, and the adjoint Schr\"odinger bridge sampler (ASBS) is a simulation-free diffusion sampler designed for exactly such unnormalized targets.
Hard equality path and terminal constraints, by contrast, confine the admissible decision variables to a measure-zero feasibility manifold, on which the target must be redefined intrinsically. 
In particular: $(a)$ we analyze two complementary parametrizations; a rollout parametrization, in which the dynamics are eliminated and only the remaining constraints shape the manifold, and a double-shooting parametrization, in which states and controls are sampled jointly and the dynamics themselves become part of the manifold; $(b)$ we establish regularity conditions under which both admissible sets are smooth embedded manifolds; $(c)$ under compactness and path-connectedness assumptions, we sample from the resulting intrinsic Gibbs measures via Riemannian ASBS, treating strict inequalities through exponential slack variables. Experiments, including contact-rich locomotion and manipulation, demonstrate the effectiveness of both regimes.
\end{abstract}
\section{Introduction} \label{sec: intro}
\IEEEPARstart{T}{his} paper concerns two problems of rather different origin. \emph{Trajectory optimization} seeks a state--control trajectory of minimal cost subject to dynamics and task constraints, and is a computational workhorse of model-based robotics. The \emph{Schr\"odinger bridge problem} \cite{leonard2014survey} seeks the controlled diffusion that most economically transports a tractable source distribution onto a prescribed target, and has recently matured into a family of practical diffusion samplers \cite{liu2023adjoint}. Our aim is to make the relation between the two precise, and then to exploit it. Concretely, consider a deterministic, nonlinear, time-varying, discrete-time dynamical system
\begin{equation}
    x_{t+1} = f_t(x_t, u_t), \quad t \in \N_{[0:T-1]},
    \label{eq: discrete_dyn}
\end{equation}
where $x_t \in \R^n$ is the state and $u_t \in \R^m$ is the control input. For a finite horizon $T$, a state-control trajectory is denoted by
\begin{equation}
    \tau = (x_0, u_0, x_1, u_1, \dots, x_{T-1}, u_{T-1}, x_T) \in \R^{n_\tau},
    \label{eq: tau}
\end{equation}
with $n_\tau := (T+1)n + Tm$.
We study the finite-horizon trajectory optimization problem, formally defined as
\begin{subequations} \label{eq: constrained_TO}
\begin{align}
    \min_{\tau} \ \ &\mathcal{J}(\tau) := \sum_{t=0}^{T-1} \ell_t (x_t, u_t) + \ell_T (x_T) \\
    \text{s.t.} \ \ &\eqref{eq: discrete_dyn}, \quad x_0 = \bar{x}_0, \\
    &h_t(x_t, u_t) = 0, \ t \in \N_{[0:T-1]}, \quad h_T(x_T) = 0, \label{eq: equality_constraints}
\end{align}
\end{subequations}
where $\ell_t(\cdot, \cdot) \ge 0$ and $\ell_T(\cdot) \ge 0$ are running and terminal costs, the smooth maps $h_t: \R^n \times \R^m \to \R^{r_t}$ and $h_T: \R^n \to \R^{r_T}$ encode, e.g., exact goal reaching, waypoint passage, actuation geometry, or safety specifications (strict inequalities are treated in Sec.~\ref{subsec: assumptions_inequalities} via slack variables). We distinguish two regimes of \eqref{eq: constrained_TO} throughout: the \emph{soft-constrained} regime, in which the maps in \eqref{eq: equality_constraints} are absent or their violation is merely penalized within $\mathcal{J}$, and the \emph{hard-constrained} regime, in which \eqref{eq: equality_constraints} must hold exactly. 

Gradient-based solvers, from DDP/iLQR to constrained direct-transcription NLP methods \cite{jacobson1970differential, li2004iterative, betts2010practical}, are local---they refine a single initial guess into one locally optimal trajectory---
although constrained formulations can enforce \eqref{eq: equality_constraints} at convergence.
Sampling-based approaches such as model predictive path integral control (MPPI) \cite{williams2017mppi} explore the control space more globally, parallelize naturally, and query the dynamics model only to evaluate trajectory costs, thereby avoiding dynamics derivatives.
These methods, however, together with more recent inference-based solvers \cite{yang2026tsmc}, handle \eqref{eq: equality_constraints} through \emph{soft} penalties, so that feasibility is traded against cost and never formally guaranteed. Two questions therefore structure this paper: \emph{what, precisely, is the relation between trajectory optimization and diffusion-based sampling}, and \emph{can the global, distributional exploration of the latter be retained while satisfying hard constraints exactly?}

The starting point is a characterization of trajectory optimization as inference \cite[Thm.~1]{yang2026tsmc}: instead of a single decision variable $\Theta$ (here, a control sequence or a full trajectory), one optimizes a distribution over decision variables, regularized by the Kullback--Leibler (KL) divergence to a prior $p_0$.
The optimal distribution admits the closed-form expression
\begin{equation}
    p^\star(\Theta) = \frac{1}{Z} p_0(\Theta) \exp \left(-\frac{1}{\lambda} E(\Theta)\right),
    \label{eq: p_star}
\end{equation}
where $E$ is an energy induced by $\mathcal{J}$ and $\lambda > 0$ a temperature, so that $p^\star$ concentrates on minimizers as $\lambda \to 0^+$. Generating draws from the unnormalized density \eqref{eq: p_star} is precisely the problem addressed by ASBS \cite{liu2023adjoint}.
In the hard-constrained regime, by contrast, the admissible decision variables form a \emph{feasibility manifold}, a measure-zero subset of the ambient space.
This motivates adopting the Riemannian Adjoint Schr\"odinger Bridge sampler (R--ASBS) \cite{mosso2026hardconstrained}
as the sampling engine for the hard-constrained case, for which we develop the supporting theory.

\parhead{Contributions.} $(i)$ We investigate the relation between trajectory optimization and Schr\"odinger bridge sampling. In rollout coordinates, soft-constrained problems reduce to Euclidean Gibbs--Boltzmann sampling solved directly by ASBS, while hard-constrained problems become Gibbs--Boltzmann sampling on feasibility manifolds (Thm.~\ref{theo: optimization_inference_manifold}) solved by R--ASBS. $(ii)$ For the hard-constrained case, we establish regularity conditions (Props.~\ref{prop: regularity_S} and~\ref{prop: regularity_Mp}) under which the corresponding feasible sets are smooth embedded manifolds; compactness and path-connectedness are imposed as standing assumptions matching the R--ASBS hypotheses; strict inequalities are handled through exponential slacks. $(iii)$ We demonstrate both regimes on diverse robotic tasks (Sec.~\ref{sec: numerics}).

\section{Related Work}
\label{app: related_work}
\parhead{Optimization-based trajectory optimization.}
Beyond the local solvers recalled in Sec.~\ref{sec: intro}, the numerical machinery for \eqref{eq: constrained_TO} comprises interior-point \cite{wachter2006implementation}, sequential-quadratic \cite{gill2005snopt}, and successive-convex \cite{mao2016successive,li25rss-crisp} programming, together with augmented-Lagrangian methods \cite{howell2019altro}, in shooting or direct-transcription form. Global alternatives, such as mixed-integer programming \cite{marcucci2024shortest} and convex relaxations \cite{kang2024fast}, offer optimality certificates at a computational cost that grows quickly with problem size.
 
\parhead{Sampling-based trajectory optimization.}
Originating in path-integral control \cite{kappen2005linear, theodorou2010generalized}, this family includes MPPI, predictive sampling \cite{howell2022predictive}, cross-entropy \cite{kobilarov2012cross} and Stein variational \cite{lambert2020stein} methods, model-based diffusion \cite{pan2024model}, and tempered sequential Monte Carlo (TSMC) \cite{yang2026tsmc}, whose inference formulation \eqref{eq: p_star} underlies Sec.~\ref{sec: manifold_TO}. These methods generally encode task constraints through costs; recent constrained Stein variants instead enforce them explicitly \cite{power2024constrained,li2026steinsqp}.
 
\parhead{Diffusion models for planning and control.}
Diffusion models trained on demonstrations synthesize trajectories \cite{janner2022diffuser}, policies \cite{chi2023diffusion}, and dynamics \cite{qi26ral-gpc}, with constraints imposed through guidance or projection steps \cite{christopher2024constrained}; their applicability therefore rests on the availability of trajectory datasets.
 
\parhead{Constrained and manifold sampling.}
Sampling on implicitly defined manifolds is classically performed with constrained Hamiltonian Monte Carlo \cite{brubaker2012family}, which is asymptotically exact but suffers from slow mixing and requires a constraint solve at every step of the chain. Mirror, reflected, and barrier constructions \cite{liu2023mirror, lou2023reflected, fishman2023diffusion} address constrained generative modeling from data.

\section{Adjoint Schr\"odinger Bridge Sampling}
\label{sec: preliminaries}
We review background on the sampling machinery of \cite{liu2023adjoint, mosso2026hardconstrained}. Throughout, $\langle \cdot, \cdot \rangle$ and $\|\cdot\|$ denote the Euclidean inner product and norm on the ambient space $\R^d$, respectively.

\subsection{Euclidean ASBS}
\label{subsec: euclidean_ASBS}
Given a continuously differentiable energy $E: \R^d \to \R$, the objective is to sample from the Gibbs--Boltzmann distribution
\begin{equation}
    \nu(\diff x) = \frac{1}{Z}\exp(-E(x))\diff x, \quad x \in \R^d,
    \label{eq: pi_Rd}
\end{equation}
without evaluating the normalizing constant $Z$. Diffusion samplers learn a control field $u_t^\theta$ driving the It\^o stochastic differential equation (SDE)
\begin{equation}
    \diff X_t = \sigma_t u_t^\theta (X_t) \diff t + \sigma_t \diff W_t, \quad X_0 \sim \mu,
    \label{eq: flat_controlled_sde}
\end{equation}
so that $X_1 \sim \nu$, where $\mu$ is a tractable source and $\sigma_t: [0,1] \to \R_+$ a noise schedule\footnote{We restrict the exposition to the zero-drift reference process.}. Let $p^u$ and $p^{\mathrm{base}}$ denote the path measures of the controlled and uncontrolled ($u \equiv 0$) processes. The Schr\"odinger bridge (SB) problem \cite{leonard2014survey} seeks the minimum-energy control---equivalently, by Girsanov's theorem, the minimum Kullback--Leibler divergence $D_{\mathrm{KL}}(p^u \Vert p^{\mathrm{base}})$---achieving the terminal constraint,
\begin{equation}
    \min_{u} \ \E_{p^u} \Big[ \int_0^1 \tfrac12 \|u_t(X_t)\|^2 \diff t \Big] \quad \text{s.t.} \quad X_1 \sim \nu.
    \label{eq: SB_flat}
\end{equation}
Notably, \eqref{eq: SB_flat} is equivalent \cite[Thm.~3.1]{liu2023adjoint} to the unconstrained terminal-cost stochastic optimal control (SOC) problem
\begin{equation}
    \min_{u} \ \E_{p^u} \Big[ \int_0^1 \tfrac12 \|u_t(X_t)\|^2 \diff t + \log \frac{\hat{\varphi}_1 (X_1)}{\nu(X_1)} \Big],
    \label{eq: flat_SB_SOC}
\end{equation}
where $\varphi_t, \hat{\varphi}_t$ are the SB potentials, satisfying $\varphi_1 \hat{\varphi}_1 = \nu$ and $\varphi_0 \hat{\varphi}_0 = \mu$, and the optimal control is $u_t^\star = \sigma_t \nabla \log \varphi_t$. ASBS \cite{liu2023adjoint} solves \eqref{eq: flat_SB_SOC} by alternating two simulation-free regressions. \emph{Adjoint matching} \cite{domingoenrich2025adjoint} exploits the additivity of the Gaussian reference kernel to characterize the optimal control as
\begin{equation}
    u^\star \!=\! \argmin_{u} \ \E \big[ \|u_t(X_t) + \sigma_t (\nabla E + \nabla \log \hat{\varphi}_1) (X_1)\|^2 \big],
\label{eq: ASBS_u_star}
\end{equation}
where the expectation is over $p_{t \vert 0,1}^{\mathrm{base}}\, p_{0,1}^{\bar{u}}$ with $\bar{u} = \texttt{stopgrad}(u)$: endpoints $(X_0, X_1)$ are drawn from the current controlled process $p_{0,1}^{\bar{u}}$, and intermediate states from the (Gaussian) base bridge $p_{t\vert 0,1}^{\mathrm{base}}$. \emph{Corrector matching} estimates the endpoint bias term $\nabla \log \hat{\varphi}_1$, induced by a general non-memoryless source $\mu$, via
\begin{equation}
    \nabla \log \hat{\varphi}_1 = \argmin_h \E_{p_{0,1}^{u}} \Big[ \Big\| h(X_1) + \frac{X_1 - X_0}{\int_0^1 \sigma_t^2 \diff t} \Big\|^2 \Big].
    \label{eq: euclidean_corrector}
\end{equation}
Alternating \eqref{eq: ASBS_u_star} and \eqref{eq: euclidean_corrector} converges to the SB solution when each stage reaches its critical point \cite[Thm.~3.2]{liu2023adjoint}. Crucially, neither differentiation through simulated trajectories nor samples from $\nu$ are required. App.~\ref{app: algorithms} summarizes the resulting procedure.

\subsection{Riemannian ASBS: Intrinsic Sampling on $\mathcal{M}$}
\label{subsec: riemannian_ASBS}
We now summarize the manifold extension \cite{mosso2026hardconstrained} of this machinery, which will form the basis of our hard-constrained formulation.
Specifically, consider
\begin{equation}
    \mathcal{M} := \{x \in \R^d : c(x) = 0\}, \qquad c: \R^d \to \R^{d_c},
    \label{eq: M_constraint}
\end{equation}
with $c$ smooth and $J_c(x) \in \R^{d_c \times d}$ of full row rank on $\mathcal{M}$, so that $\mathcal{M}$ is a smooth, boundaryless embedded submanifold of dimension $d - d_c$. We further assume $\mathcal{M}$ is compact and path-connected. Let $T_x\mathcal{M}$ denote the tangent space at $x$, $P_x := I_d - J_c(x)^\T (J_c J_c^\T)^{-1} J_c(x)$ the orthogonal projector onto $T_x\mathcal{M}$, $\nabla_{\mathcal{M}}$ the Riemannian gradient, and $\diff \mathrm{vol}_{\mathcal{M}}$ the volume form induced by the ambient metric. The intrinsic target is
\begin{equation}
    \nu_{\mathcal{M}}(\diff x) = \frac{1}{Z_{\mathcal{M}}} \exp(-E(x)) \diff \mathrm{vol}_{\mathcal{M}}(x), \quad x \in \mathcal{M}.
    \label{eq: pi_M}
\end{equation}
Define the controlled diffusion \emph{intrinsically} on $\mathcal{M}$,
\begin{equation}
\begin{split}
    \quad \diff X_t = \sigma_t &u_t(X_t) \diff t + \sigma_t \diff W_t^{\mathcal{M}}, \\ 
    X_0 \sim &\mu_{\mathcal{M}}, \ u_t(x) \in T_x \mathcal{M},
    \label{eq: manifold_controlled_sde}
\end{split}
\end{equation}
where $W_t^{\mathcal{M}}$ is Brownian motion on $\mathcal{M}$ (with generator $\frac12 \Delta_{\mathcal{M}}$) and $\mu_{\mathcal{M}}$ an easy-to-sample probability distribution on $\mathcal{M}$. Since drift and noise act through the tangent bundle, feasibility is enforced at the level of the state space ($X_0 \in \mathcal{M} \implies X_t \in \mathcal{M}$ for all $t$, almost surely).

The ASBS structure transfers to $\mathcal{M}$ \cite[Thm.~1]{mosso2026hardconstrained}: the corrector $\nabla_{\mathcal{M}} \log \hat{\varphi}_1^{\mathcal{M}}$ again removes the bias of a non-memoryless source, but is now an intrinsic tangent field.
The exact manifold analogue of \eqref{eq: ASBS_u_star} is more delicate. On a curved space, differentiating the diffusion semigroup replaces the identity transport of $\R^d$ with the adjoint of a damped stochastic parallel transport $\mathcal{W}_{t,1}$ along reference paths.

Since heat kernels, reference bridges, and damped transports are intractable on general manifolds, R--ASBS replaces them with computable geometric surrogates: a short-time Varadhan/Log-map approximation of the heat-kernel score, a noisy-geodesic approximation of the reference bridge, and Levi--Civita parallel transport along the sampled bridge.

However, the manifolds arising in trajectory optimization are implicitly defined by nonlinear equality constraints, for which those primitives are unavailable. \emph{Extended R--ASBS}~\cite[Alg.~2]{mosso2026hardconstrained} replaces them with three projection-based operations\footnote{Owing to these geometric approximations, formal convergence guarantees for R--ASBS are not available.}~: $(i)$ the retraction is the nearest-point projection $\Pi_{\mathcal{M}}$ onto $c(x) = 0$, computed by Newton iterations;
$(ii)$ backward parallel transport of the terminal adjoint is simplified by projection-as-transport (PAT), the sequential composition of tangent projectors along the discretized path, $\mathcal{T}_{1 \to t_0}(v) \approx P_{X_{t_0}} P_{X_{t_1}} \cdots P_{X_{t_{N-1}}} v$; $(iii)$ the corrector target is the projected chord $-P_{X_1}(X_1 - X_0)/\int_0^1 \sigma_t^2 \diff t$, i.e., the tangent projection of the Euclidean corrector in \eqref{eq: euclidean_corrector}. The complete procedure is given in Alg.~\ref{alg: rasbs}, App.~\ref{app: algorithms}.

\section{Trajectory Optimization via Schr\"odinger Bridge Sampling}
\label{sec: manifold_TO}
This section makes the relation announced in Sec.~\ref{sec: intro} clear. Sec.~\ref{subsec: rollout_coordinates} introduces rollout coordinates and shows that, in the soft-constrained regime, \eqref{eq: constrained_TO} reduces to a Euclidean Gibbs--Boltzmann sampling problem solved verbatim by ASBS. The remainder of the section develops the hard-constrained regime, in which the decision variable is confined to a manifold of the form \eqref{eq: M_constraint}.
All formal proofs are collected, with additional remarks, in App.~\ref{app: proofs}.

\subsection{Rollout Coordinates and the Soft-Constrained Case}
\label{subsec: rollout_coordinates}
With the initial state fixed at $\bar{x}_0$, the natural minimal decision variable
is the open-loop control sequence $z := (u_0, \dots, u_{T-1}) \in \R^{Tm}$. The dynamics are eliminated through the rollout map $\Phi: \R^{Tm} \to \R^{n_\tau}$, which recursively sets $x_{t+1} = f_t(x_t, u_t)$ from $x_0 = \bar{x}_0$ and returns the full trajectory $\tau = \Phi(z)$.

In the soft-constrained regime ($h_t, h_T$ absent, or penalized within the cost), the decision variable is fully Euclidean, and the KL-regularized problem behind \eqref{eq: p_star} takes, for a strictly positive prior density $p_0$ on $\R^{Tm}$, the explicit form
\begin{equation}
    \min_{p \in \mathcal{P}(\R^{Tm})} \ \E_{z \sim p}\left[\mathcal{J}(\Phi(z))\right] + \lambda\, D_{\mathrm{KL}}(p \Vert p_0),
    \label{eq: KL_euclidean}
\end{equation}
with $\mathcal{P}(\R^{Tm})$ the Borel probability measures on $\R^{Tm}$. Since $\mathcal{J} \ge 0$ and $p_0$ is a strictly positive density, the partition function $Z_\lambda := \int_{\R^{Tm}} p_0(z)\, e^{-\frac{1}{\lambda} \mathcal{J}(\Phi(z))}\, \diff z \in (0, 1]$, and \eqref{eq: KL_euclidean} admits the unique minimizer \cite[Thm.~1]{yang2026tsmc}
\begin{equation}
    p^\star_\lambda(z) = \frac{1}{Z_\lambda}\, p_0(z)\, e^{-\frac{1}{\lambda} \mathcal{J}(\Phi(z))},
    \label{eq: euclidean_TO_target}
\end{equation}
the instance of \eqref{eq: p_star} with $\Theta = z$. Setting
\begin{equation}
    \tilde{E}(z) := \frac{1}{\lambda}\, \mathcal{J}(\Phi(z)) - \log p_0(z)
    \label{eq: euclidean_TO_energy}
\end{equation}
identifies \eqref{eq: euclidean_TO_target} with the unnormalized target \eqref{eq: pi_Rd}, i.e., trajectory optimization reduces, with no further structure, to the sampling problem of Sec.~\ref{subsec: euclidean_ASBS}.
ASBS alternates \eqref{eq: ASBS_u_star} and \eqref{eq: euclidean_corrector}, touching the dynamics only through rollout derivatives of $\tilde{E}$ in the adjoint-matching target\footnote{A surrogate gradient serves where the simulator is not differentiable.}, and its samples concentrate on low-cost control sequences as $\lambda \to 0^+$.

\subsection{Hard Constraints}
\label{subsec: hard_rollout}
Hard constraints change the geometry. Stacking the path and terminal maps along the rollout defines $G: \R^{Tm} \to \R^{R}$,
\begin{equation}
G(z) =
\begin{bmatrix}
h_0(\bar{x}_0, u_0)\\
\vdots\\
h_{T-1}(x_{T-1}(z), u_{T-1})\\
h_T(x_T(z))
\end{bmatrix},
\label{eq: G_def}
\end{equation}
with $R := \sum_{t=0}^{T-1} r_t + r_T$, and the feasible set in rollout coordinates is $\mathcal{X} := G^{-1}(0) \subset \R^{Tm}$. Note that $G$ contains no dynamics rows; the dynamics are satisfied by construction, and only the remaining constraints shape the manifold.

\begin{proposition}[Regularity of the rollout feasible set]
\label{prop: regularity_S}
Let $f_t$, $h_t$ ($t = 0, \dots, T-1$), and $h_T$ be smooth. Assume $\mathcal{X} \neq \emptyset$ and $\operatorname{rank} J_G(z) = R$ for all $z \in \mathcal{X}$. Then $\mathcal{X}$ is a smooth embedded submanifold of $\R^{Tm}$ of dimension $Tm - R$.
\end{proposition}
The rank requirement is not automatic\footnote{Inconsistent, redundant, or singular path and terminal constraints may violate it.} and couples all stages through the rollout; a more tractable stagewise sufficient condition is given in Cor.~\ref{cor: sufficient_regular_terminal}, App.~\ref{app: proofs}.

Since $\mathcal{X}$ has codimension $R \ge 1$, it is a Lebesgue-null subset of $\R^{Tm}$
The KL-regularized problem must therefore be re-posed on $\mathcal{M}$. In the Euclidean problem \eqref{eq: KL_euclidean}, the structural role of the full-support prior $p_0$ is to render the tilted density normalizable on the ambient space. On a compact manifold (cf.\ Asm.~\ref{ass: compact_connected}) this role is superfluous and we regularize directly toward the uniform law,
\begin{equation}
    \min_{p \in \mathcal{P}(\mathcal{M})} \ \E_{x \sim p}[E(x)] + \lambda\, D_{\mathrm{KL}}(p \Vert \bar{\mu}_{\mathcal{M}}),
    \label{eq: KL_manifold}
\end{equation}
where $\mathcal{M}$ is a generic manifold of the form \eqref{eq: M_constraint} and $\bar{\mu}_{\mathcal{M}} := \mathrm{vol}_{\mathcal{M}}(\mathcal{M})^{-1}\, \mathrm{vol}_{\mathcal{M}}$ the uniform law. We then have the following result.

\begin{theorem}[Optimization as inference on $\mathcal{M}$]
\label{theo: optimization_inference_manifold}
Let $\mathcal{M} \subset \R^d$ be a smooth, compact, boundaryless embedded submanifold, with volume measure $\mathrm{vol}_{\mathcal{M}}$ induced by the ambient metric and finite total volume $V_{\mathcal{M}} := \mathrm{vol}_{\mathcal{M}}(\mathcal{M}) \in (0, \infty)$. Let $E: \mathcal{M} \to \R$, $E \in C^1(\mathcal{M})$ denote the trajectory cost expressed in the chosen coordinates, and $\lambda > 0$. Then:
\begin{enumerate}
    \item[(i)] \emph{(Well-posedness.)} The partition function $Z_{\mathcal{M}} := \int_{\mathcal{M}} \exp(-\frac{1}{\lambda} E(x))\, \diff\mathrm{vol}_{\mathcal{M}}(x)$ satisfies $0 < Z_{\mathcal{M}} < \infty$.
    \item[(ii)] \emph{(Gibbs--Boltzmann form.)} Problem \eqref{eq: KL_manifold} has the unique minimizer
    \begin{equation}
        \nu_{\mathcal{M}}(\diff x) = \frac{1}{Z_{\mathcal{M}}}\, e^{-\frac{1}{\lambda} E(x)}\, \diff\mathrm{vol}_{\mathcal{M}}(x).
        \label{eq: nu_star_manifold}
    \end{equation}
    \item[(iii)] \emph{(Low-temperature concentration.)} Let $\mathcal{M}^\star := \argmin_{x \in \mathcal{M}} E(x)$, which is nonempty by compactness. For every open set $U \supseteq \mathcal{M}^\star, \ \nu_{\mathcal{M}}(\mathcal{M} \setminus U) \to 0$ as $\lambda \to 0^+$.
\end{enumerate}
\end{theorem}

We equip $\mathcal{X}$ with the Riemannian metric inherited from the ambient rollout-coordinate Euclidean space and denote the induced volume form by $\mathrm{dvol}_{\mathcal{X}}$. Constrained trajectory optimization then becomes sampling from the intrinsic Gibbs measure on $\mathcal{X}$,
\begin{equation}
    p_{\mathcal{X}}^{\star}(\diff z) = \frac{1}{Z_{\mathcal{X}}}  e^{-\frac{1}{\lambda} \mathcal{J}(\Phi(z))} \mathrm{dvol}_{\mathcal{X}}(z),
    \label{eq: p_S_star}
\end{equation}
which is the instance of \eqref{eq: nu_star_manifold} for $\mathcal{M} = \mathcal{X}$ and $E = \mathcal{J} \circ \Phi$.
Under Asms.~\ref{ass: compact_connected} and~\ref{ass: sampling_access} (Sec.~\ref{subsec: assumptions_inequalities}), Extended R--ASBS applies directly with $c := G$ and manifold energy $E_{\mathcal{X}}(z) = \frac{1}{\lambda} \mathcal{J}(\Phi(z))$; samples $z \sim p_{\mathcal{X}}^\star$ are mapped to feasible trajectories by $\tau = \Phi(z)$.

\parhead{Advantages and limitations.}
Rollout coordinates realize the smallest possible sampling dimension ($Tm$, codimension $R$ only), and dynamic feasibility can never be violated, not even transiently during Newton projections. The price is that both the energy and the constraint map \eqref{eq: G_def} are $T$-fold compositions of the dynamics. Their derivatives chain the stage Jacobians $A_t := \partial_x f_t$, $B_t := \partial_u f_t$: for $s < t$, $\frac{\partial x_t}{\partial u_s}(z) = A_{t-1} A_{t-2} \cdots A_{s+1} B_s$, so gradient magnitudes can grow (or vanish) exponentially with the horizon whenever the linearized dynamics are unstable or stiff \cite{pascanu2013difficulty}.

\subsection{Double-Shooting Coordinates: Dynamics as Manifold Constraints}
\label{subsec: double_shooting}
The complementary formulation samples states and controls \emph{jointly}. The decision variable is the full trajectory \eqref{eq: tau}, and the dynamics enter as explicit equality constraints, defining the dynamics manifold $\mathcal{M}_d := \{\tau \in \R^{n_\tau} : \eqref{eq: discrete_dyn} \text{ holds } \forall t \in \N_{[0:T-1]}\}$.
Adding the initial condition and the path and terminal equality constraints of \eqref{eq: constrained_TO} yields the constrained trajectory manifold
\begin{equation}
\begin{split}
    \mathcal{M}_p := \{&\tau \in \mathcal{M}_d : x_0 = \bar{x}_0, \\ 
    &h_t(x_t, u_t) = 0 \ \forall t, \ h_T(x_T) = 0\}.
\end{split}
\label{eq: Mp_def}
\end{equation}
The next proposition formalizes that $\mathcal{M}_d$ is globally parametrized by the initial state and the control sequence and that the regularity of $\mathcal{M}_p$ is inherited from that of $\mathcal{X}$.

\begin{proposition}[Rollout embedding and regularity of $\mathcal{M}_d$ and $\mathcal{M}_p$]
\label{prop: regularity_Mp}
Let $f_t$, $h_t$ ($t = 0, \dots, T-1$), and $h_T$ be smooth, and extend the rollout map to $\Phi: \R^{n+Tm} \to \R^{n_\tau}$ by treating $x_0$ as free. Then:
\begin{enumerate}
    \item[(i)] $\Phi$ is a smooth embedding with $\Phi(\R^{n+Tm}) = \mathcal{M}_d$. In particular, $\mathcal{M}_d$ is a smooth embedded submanifold of $\R^{n_\tau}$ of dimension $n + Tm$, and no rank assumption is required: the defect constraints are regular everywhere.
    \item[(ii)] If, in addition, $\mathcal{X} \neq \emptyset$ and $\operatorname{rank} J_G(z) = R$ for all $z \in \mathcal{X}$ (the assumptions of Prop.~\ref{prop: regularity_S}), then $\mathcal{M}_p = \Phi(\{\bar{x}_0\} \times \mathcal{X})$ is a smooth embedded submanifold of $\R^{n_\tau}$, diffeomorphic to $\mathcal{X}$, of dimension $Tm - R$.
\end{enumerate}
\end{proposition}

\parhead{Advantages and limitations.}
In double-shooting coordinates, the constraint map stacking $(c_t)_t$ and $(h_t, h_T)$ is local in time, i.e., each row involves only variables from adjacent stages. The energy gradient $\nabla_\tau \mathcal{J}$ decomposes stagewise, curing the exploding-gradient pathology. The downside is sampling hardness, since the ambient dimension inflates to $n_\tau$ and the manifold codimension to $(T+1)n + R$, so every step must retract onto a large nonlinear manifold. 

\subsection{Standing Assumptions and Strict Inequalities}
\label{subsec: assumptions_inequalities}
R--ASBS requires compact, path-connected manifolds. Smoothness and embeddedness are guaranteed by Props.~\ref{prop: regularity_S} and~\ref{prop: regularity_Mp}; since $\mathcal{X}$ and $\mathcal{M}_p$ are diffeomorphic, compactness and path-connectedness hold for one iff they hold for the other, but they cannot, in general, be derived from the problem data and are therefore imposed as a standing assumption. 

\begin{assumption}[Compactness and path-connectedness]
\label{ass: compact_connected}
The feasible set $\mathcal{X}$ (equivalently, $\mathcal{M}_p$) is nonempty, compact, and path-connected.
\end{assumption}
Closedness is automatic ($\mathcal{X}$ is the preimage of $\{0\}$ under the continuous map $G$), so compactness reduces to boundedness. Boundedness holds, for instance, when the decision variable ranges over a compact ambient reparametrization, e.g., angles identified on the torus via unit-vector coordinates.

R--ASBS additionally requires an initial law for the intrinsic diffusion \eqref{eq: manifold_controlled_sde}.
\begin{assumption}[Sampling access to a feasible source]
\label{ass: sampling_access}
We have sampling access to a source probability measure $\mu_{\mathcal{M}}$ supported on $\mathcal{X}$ (equivalently, $\mathcal{M}_p$).
\end{assumption}
In practice this requires only an initial trajectory close enough to feasibility for the nearest-point projection $\Pi_\mathcal{M}$ to converge.
Starting from this feasible but suboptimal source, R--ASBS transports mass to $p_{\mathcal{X}}^\star$ (equivalently $p_{{\mathcal{M}}_p}^\star$), whose samples remain feasible and, as $\lambda \to 0^+$, concentrate on the trajectories that are optimal with respect to the soft objectives collected in $\mathcal{J}$ (Thm.~\ref{theo: optimization_inference_manifold}(iii)).

\parhead{Strict inequality constraints.}
Strict path inequalities $g_t(x_t, u_t) < 0$ can be encoded by introducing slack variables $s_t \in \R^{q_t}$ and imposing $g_t(x_t, u_t) + \exp(s_t) = 0$ componentwise. Since $\partial_{s_t}\big(g_t + \exp(s_t)\big) = \operatorname{diag}(\exp(s_t))$ is nonsingular for every finite $s_t$, the augmented constraints are regular in the slack directions, and the joint constraint map retains full row rank whenever the equality block does. Thus, at the level of manifold regularity, strict inequalities reduce to equality constraints on an augmented space.
Observe that this construction represents only the strict interior $g_t < 0$, not the closed set $g_t \le 0$. Boundary-active solutions correspond to $s_t \to -\infty$, which breaks the boundedness required by Asm.~\ref{ass: compact_connected} and causes numerical ill-conditioning. We therefore target problems whose optima lie strictly in the interior of the feasible set, as in Sec.~\ref{subsec: underwater}.

\section{Numerical Experiments}
\label{sec: numerics}
Sec.~\ref{sec: manifold_TO} identifies two regimes: a Euclidean one, in which soft-constrained trajectory optimization is directly amenable to ASBS, and a hard-constrained one, in which the intrinsic target \eqref{eq: nu_star_manifold} is provably well-posed and the R--ASBS hypotheses hold. This section probes their  validity empirically beyond it.
All simulations were implemented in Python using JAX \cite{jax2018github} on a Lambda workstation with an AMD Ryzen Threadripper PRO 5975WX (32 cores) and two NVIDIA Ada6000 GPUs. The source code is available at \href{https://github.com/lamb97/R_ASBS_TO}{Github}.

\subsection{Soft-Constrained Tasks via Euclidean ASBS}
\label{subsec: euclidean_tasks}
\begin{figure}[t]
    \centering
    \includegraphics[width=\linewidth]{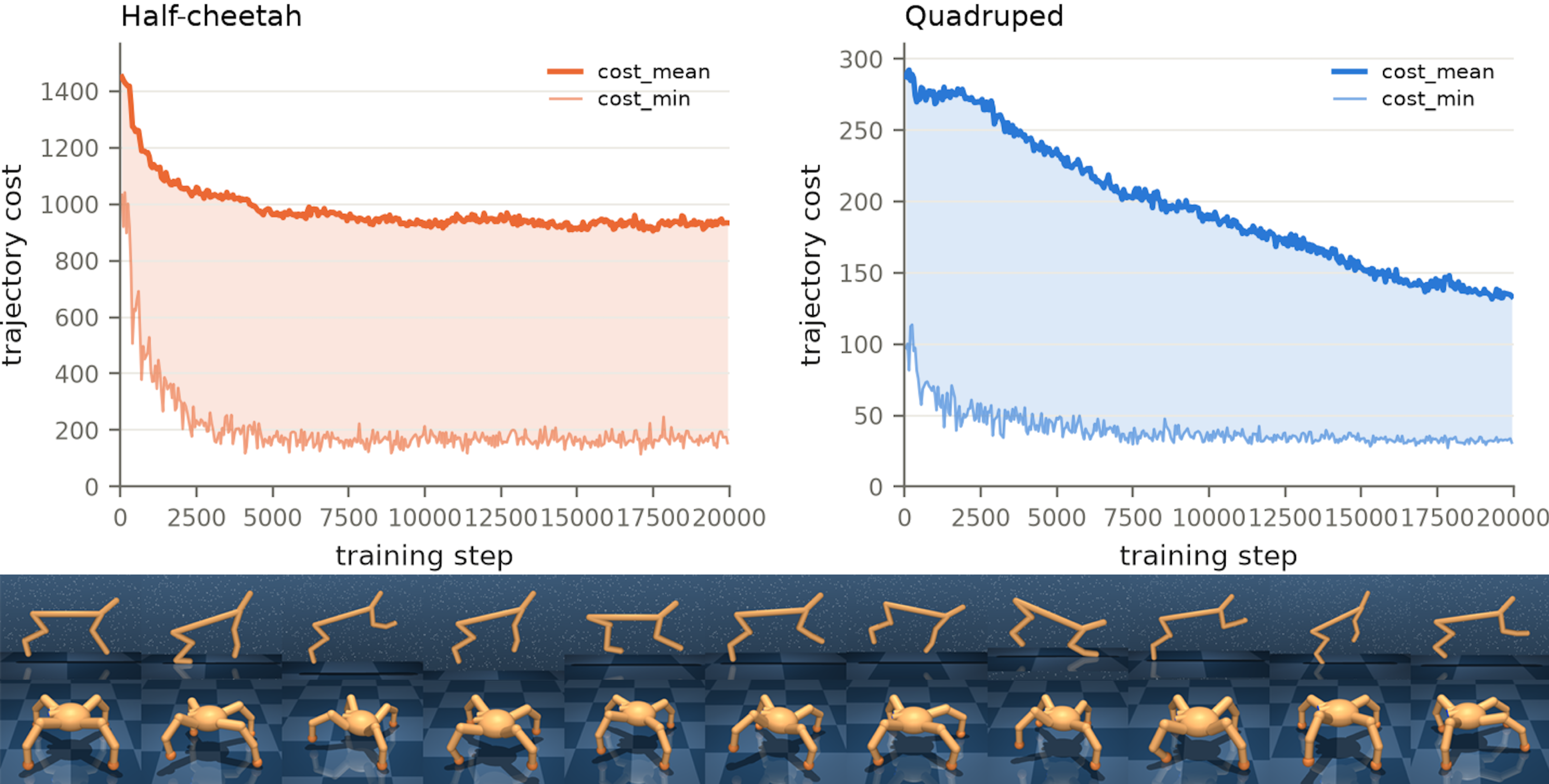}
    \caption{Cost of the sampled trajectories during training: mean over the $2000$ draws at each evaluation (thick) and lowest-cost draw (thin). The bottom rows show time steps from the corresponding minimum-cost trajectories.}
    \label{fig: cost_distribution}
    \vspace{-5mm}
\end{figure}

Both tasks are simulated with the analytical contact engine ComFree \cite{borse2026comfree}, whose compliant contact model keeps every integration step differentiable: the energy gradient $\nabla_z \mathcal{J}$ is obtained by reverse-mode differentiation through the entire rollout, rather than by finite differences or zeroth-order estimation, and agrees with central finite differences across four decades of perturbation size. Fig.~\ref{fig: cost_distribution} reports $2000$ sampled trajectories per evaluation, together with the selected lowest-cost trajectory $z^\star := \argmin_{z^{(i)}} \mathcal{J}\big(z^{(i)}\big)$, $\ i = 1, \dots, 2000$.

\subsubsection{Cheetah Locomotion}
\label{subsubsec: cheetah}
We first consider the planar half-cheetah of the DeepMind Control Suite \cite{tunyasuvunakool2020dmcontrol},
on a horizon of $T = 150$ control knots, each held for $10$ integration steps, the sampled variable is $z \in \R^{900}$ while every rollout resolves $1500$ contact steps. The running cost penalizes the deviation of the forward speed from the benchmark target through its smoothed reward, together with small control-effort and control-smoothness terms, and is accumulated at every integration step.

\subsubsection{Quadruped Locomotion}
\label{subsubsec: quadruped}
The second task is the spatial quadruped of the same suite \cite{tunyasuvunakool2020dmcontrol}, whose $m = 12$ actuators drive four legs that alternate stance and flight.
With the same $T=150$ knots, each now held for $4$ integration steps, the sampled variable is $z \in \R^{1800}$.
The running cost is the benchmark Move objective.

\subsection{Hard Equality Constraints via R--ASBS}
\label{subsec: equality_tasks}

\subsubsection{$2$D PointMaze}
We evaluate the R--ASBS trajectory sampler on a planar PointMaze navigation problem inspired by the OGBench benchmark \cite{park2025ogbench}. A two-dimensional point mass must move from a fixed initial state $x_0 = \bar{x}_0$ to a prescribed terminal location under the dynamics $x_{t+1} = x_t + \Delta t\, u_t$. Two classes of hard equality constraints are imposed: exact terminal arrival, $x_T = x_T^\star$, and constant speed along the path, $\|u_t\|^2 = v_0^2$ for all $t$. 
The maze itself is encoded through the running cost. Fig.~\ref{fig: const_vel_maze} shows $500$ sampled trajectories (blue), together with the selected lowest-cost trajectory $\theta^\star := \argmin_{\theta^{(i)}} \mathcal{J}\big(\Phi(\theta^{(i)})\big)$, $\ i = 1, \dots, 500$.
We compare against constrained nonlinear programming (NLP; \texttt{scipy.optimize.minimize}) and Riemannian gradient descent (RGD) \cite{boumal2023optimization}, both run from the same prior and in double-shooting coordinates.
\begin{figure}[t]
    \centering
    \includegraphics[width=0.72\linewidth]{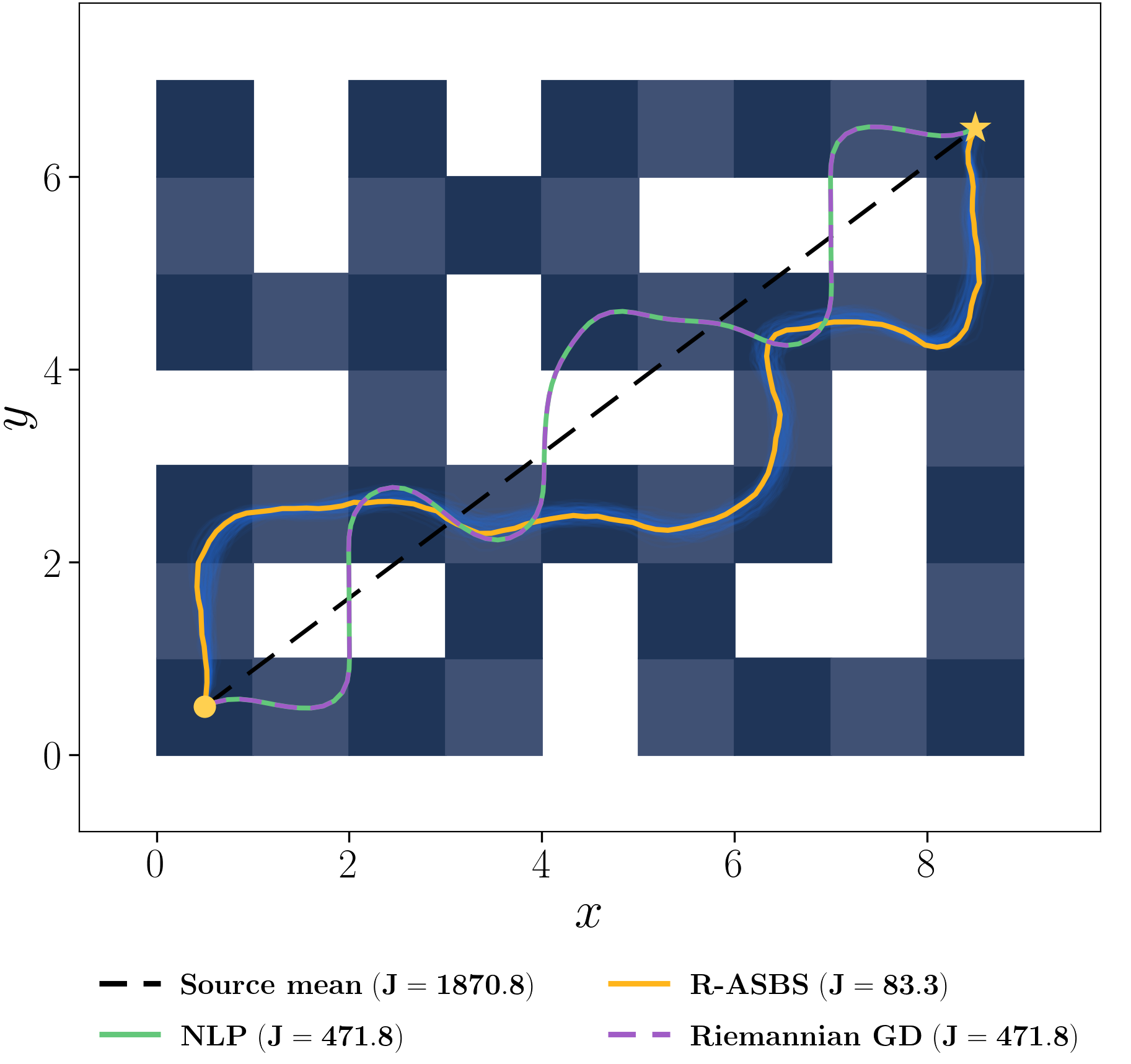}
    \caption{Constant-speed PointMaze, with walls in white. The marks $\textcolor{yellow!50!orange}{\bullet}$ and $\color{yellow!80!black}\star$ denote the initial state and the hard terminal goal.}
    \label{fig: const_vel_maze}
\end{figure}

\begin{figure}[t]
    \centering
    \vspace{-3mm}
    \includegraphics[width=\linewidth]{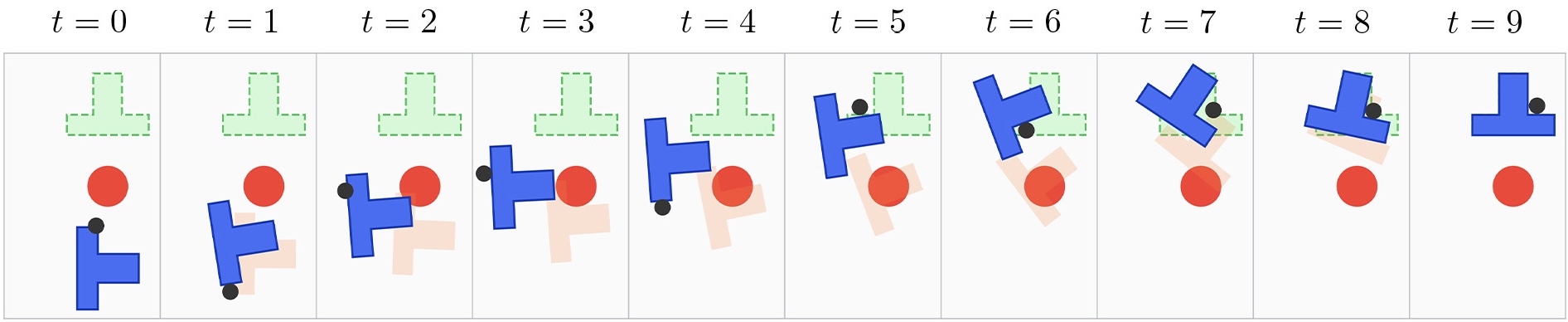}
    \caption{Push-T at successive time steps (prior $\textcolor{white!60!orange}{\blacksquare}$, best $\textcolor{white!30!blue}{\blacksquare}$, obstacle $\textcolor{red}{\bullet}$).}
    \label{fig: push_T}
    \vspace{-5mm}
\end{figure}

\subsubsection{Multi-Drone Cooperative Transport}
\label{subsubsec: multi_drone}
The complexity increases substantially in the scenario of Fig.~\ref{fig: multi_agent}$(a)$, where a team of $M = 3$ unmanned aerial vehicles (UAVs) cooperatively carries a cable-suspended payload with position and velocity $p_t^L, v_t^L \in \R^3$. The payload must move from a starting configuration to a final target, passing through required intermediate waypoints, while avoiding collisions with moving obstacles and between drones, low-altitude flight, and excessive payload swing.

Rather than sampling UAV positions and enforcing the fixed cable lengths $l_i$ as constraints, each cable direction is parametrized by a polar angle $\vartheta_t^i$ and an azimuth angle $\psi_t^i$, so that UAV $i$'s position is reconstructed as $p_t^i = p_t^L + l_i q_t^i(\vartheta_t^i, \psi_t^i)$ with $q_t^i$ the corresponding unit vector. 
The state collects payload position/velocity, cable angles, and angular rates ($n = 18$); the control $u_t \in \R^{3 + 2M}$ commands the payload acceleration and the cable angular accelerations, with magnitude limits enforced via $\tanh$. The sampled variable is $u_{[0:T-1]}$, and the trajectory is propagated through nonlinear time-varying dynamics that couple payload swing, cable geometry, and a position- and time-dependent wind field. The hard equality constraints comprise the full initial and terminal state, and passage of the payload through two waypoints at prescribed times $t_1$ and $t_2$; the remaining specifications enter as soft running costs. Fig.~\ref{fig: multi_agent}$(b)$ reports the selected lowest-cost trajectory out of $2500$ sampled ones. 

\begin{figure*}[t]
    \centering
    \includegraphics[width=0.87\textwidth]{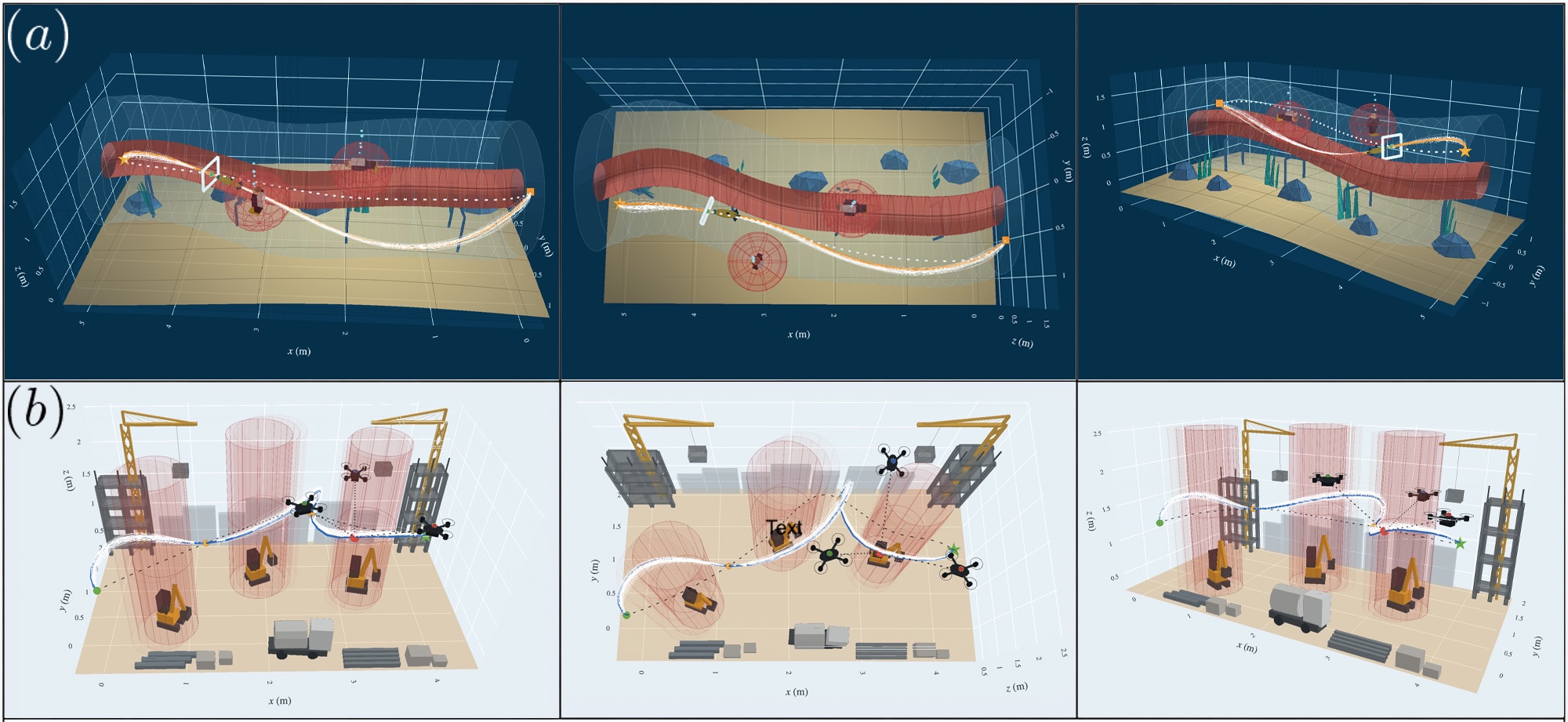}
    \caption{$(a)$ AUV pipeline inspection and $(b)$ Multi-drone cooperative payload transport, each in front, top, and side views. 
    Each panel also shows a batch of sampled feasible trajectories and highlights the mean of the feasible source distribution (dashed line).}
    \label{fig: multi_agent}
    \vspace{-5mm}
\end{figure*}

\subsubsection{Push-T Task}
\label{subsec: pushT}
Inspired by the contact-mode formulation in \cite{kang2025global}, we model the planar motion of a T-shaped slider under prescribed pusher contact.
The state is
$x_t = [b_{x,t}, b_{y,t}, \alpha_t] \in \mathbb{R}^3$,
and each control
$u_t = [e_t, s_t, v_{x,t}, v_{y,t}] \in \mathbb{R}^4$
specifies a boundary segment, a contact location, and a bounded local contact velocity.
We use a soft Gumbel--Softmax relaxation for the segment selector and compute the slider motion using a quasi-static sticking-contact model.
For a horizon of $T=200$, the target pose is imposed through three hard terminal equality constraints, while control smoothness and obstacle avoidance are treated as soft costs (Fig.~\ref{fig: push_T}), which reports the selected lowest-cost trajectory out of $500$ sampled ones.

\subsection{Strict Inequality Constraints}
\label{subsec: underwater}
We next evaluate the slack construction introduced in Sec.~\ref{subsec: assumptions_inequalities}. 
An autonomous underwater vehicle (AUV) inspects a subsea pipeline and must move from a prescribed initial condition to a terminal target while passing through one inspection gate, remaining inside a safe inspection annulus around the pipeline, avoiding a central adverse-flow region, staying above the seafloor, and maintaining a safety distance from two divers working near the pipe (Fig.~\ref{fig: multi_agent}$(b)$).

The AUV state is $x_t = [p_t, v_t]^\T \in \R^6$ (inertial position and water-relative velocity), and the sampled control is the open-loop raw acceleration sequence $u \in \R^{3T}$ with bounded physical acceleration $a_t^{\mathrm{cmd}} = a_{\max} \tanh(u_t)$; the rollout dynamics include a position- and time-dependent ocean current with an adverse pipe-core component, nonlinear hydrodynamic drag, and a small vertical bias. The strict annular inspection region is described by two inequalities,
\begin{align*}
    g_{\mathrm{out}}(p_t) &= d_{\mathrm{pipe}}^2(p_t) - R_{\max}^2 < 0, \\
    g_{\mathrm{in}}(p_t) &= R_{\min}^2 - d_{\mathrm{pipe}}^2(p_t) < 0,
\end{align*}
where $d_{\mathrm{pipe}}(p_t)$ is the distance to the pipe axis; seafloor clearance is imposed through $g_{\mathrm{floor}}(p_t) = h_{\mathrm{floor}}(p_t^x, p_t^y) + h_{\min} - p_t^z < 0$ for a smooth height field $h_{\mathrm{floor}}$, and diver avoidance through $g_{\mathrm{div},k}(p_t) = r_k^2 - \|p_t - o_k\|^2 < 0$ for fixed work zones centered at $o_k$, $k = 1, 2$. Fig.~\ref{fig: multi_agent}$(a)$ reports the selected lowest-cost trajectory out of $2500$ sampled ones. 

\section{Conclusion}
\label{sec: conclusion}
We studied the connection between finite-horizon trajectory optimization and diffusion Schr\"odinger bridge sampling, and subsequently presented a sampling-based framework for hard-constrained instances.
Numerical experiments demonstrated the efficacy of the proposed method.
A limitation, left for future work, is that the computational complexity increases with the time horizon, or more precisely with the number of time steps $T$.

\section*{Acknowledgment}
The authors would like to thank Shucheng Kang for his helpful discussions and comments.
They also express their gratitude to Luigi Tesio for his assistance in designing Fig.~\ref{fig: multi_agent}.

\appendices

\section{Supplementary Materials}
\label{app: proofs}

\begin{proof}[Proof of Thm.~\ref{theo: optimization_inference_manifold}]
Since $\mathcal{M}$ is compact and $E$ is continuous, $E$ attains a finite minimum $E^\star := \min_{\mathcal{M}} E$ and maximum $E^{\max} := \max_{\mathcal{M}} E$. Hence $e^{-E^{\max}/\lambda} \le e^{-E(x)/\lambda} \le e^{-E^\star/\lambda}$ for all $x \in \mathcal{M}$, and integrating against $\mathrm{vol}_{\mathcal{M}}$ yields 
\begin{equation*}
   V_{\mathcal{M}}\, e^{-E^{\max}/\lambda} \le Z_{\mathcal{M}} \le V_{\mathcal{M}}\, e^{-E^\star/\lambda} \implies 0 < Z_{\mathcal{M}} < \infty.
\end{equation*}
Equivalently, $\tilde{Z} := Z_{\mathcal{M}} / V_{\mathcal{M}} = \int_{\mathcal{M}} e^{-E/\lambda}\, \diff\bar{\mu}_{\mathcal{M}} \in (0, \infty)$, and $\nu_{\mathcal{M}}$ in \eqref{eq: nu_star_manifold} has bounded, strictly positive density $\tfrac{\diff \nu_{\mathcal{M}}}{\diff \bar{\mu}_{\mathcal{M}}} = \tilde{Z}^{-1} e^{-E/\lambda}$ with respect to $\bar{\mu}_{\mathcal{M}}$; in particular $\nu_{\mathcal{M}}$ and $\bar{\mu}_{\mathcal{M}}$ are mutually absolutely continuous.

For (ii), observe that $\E_{x \sim p}[E(x)]$ is finite for every $p \in \mathcal{P}(\mathcal{M})$ (as $E$ is bounded), whereas $D_{\mathrm{KL}}(p \Vert \bar{\mu}_{\mathcal{M}}) = +\infty$ unless $p \ll \bar{\mu}_{\mathcal{M}}$; the value $p = \bar{\mu}_{\mathcal{M}}$ attains a finite objective, so it suffices to minimize over $p \ll \bar{\mu}_{\mathcal{M}}$, for which also $p \ll \nu_{\mathcal{M}}$. By the chain rule for Radon--Nikodym derivatives, $\tfrac{\diff p}{\diff \nu_{\mathcal{M}}} = \tfrac{\diff p}{\diff \bar{\mu}_{\mathcal{M}}}\, \tilde{Z}\, e^{E/\lambda}$, hence
\begin{align*}
    D_{\mathrm{KL}}(p \Vert \nu_{\mathcal{M}})
    &= \int_{\mathcal{M}} \log \frac{\diff p}{\diff \nu_{\mathcal{M}}}\, \diff p \\
    &= \int_{\mathcal{M}} \Big( \log \tfrac{\diff p}{\diff \bar{\mu}_{\mathcal{M}}} + \log \tilde{Z} + \tfrac{1}{\lambda} E \Big)\, \diff p \\
    &= D_{\mathrm{KL}}(p \Vert \bar{\mu}_{\mathcal{M}}) + \log \tilde{Z} + \tfrac{1}{\lambda}\, \E_{p}[E].
\end{align*}
Multiplying by $\lambda$ and rearranging expresses the objective of \eqref{eq: KL_manifold} as $\E_{p}[E] + \lambda D_{\mathrm{KL}}(p \Vert \bar{\mu}_{\mathcal{M}}) = \lambda D_{\mathrm{KL}}(p \Vert \nu_{\mathcal{M}}) - \lambda \log \tilde{Z}$.

Because $D_{\mathrm{KL}}(p \Vert \nu_{\mathcal{M}}) \ge 0$ with equality iff $p = \nu_{\mathcal{M}}$, the objective is minimized uniquely at $p = \nu_{\mathcal{M}}$, with value $-\lambda \log \tilde{Z} = -\lambda \log(Z_{\mathcal{M}} / V_{\mathcal{M}})$, proving (ii).

For (iii), fix an open $U \supseteq \mathcal{M}^\star$. As $\mathcal{M} \setminus U$ is compact and $E > E^\star$ on it, $E_U := \min_{\mathcal{M} \setminus U} E > E^\star$; set $\delta := (E_U - E^\star)/2 > 0$. The sublevel set $B := \{x \in \mathcal{M} : E(x) < E^\star + \delta\}$ is open and nonempty (it contains $\mathcal{M}^\star$), hence $\mathrm{vol}_{\mathcal{M}}(B) > 0$. Using $e^{-E/\lambda} \le e^{-E_U/\lambda}$ on $\mathcal{M} \setminus U$, and $Z_{\mathcal{M}} \ge \int_B e^{-E/\lambda}\, \diff\mathrm{vol}_{\mathcal{M}} \ge \mathrm{vol}_{\mathcal{M}}(B)\, e^{-(E^\star + \delta)/\lambda}$,
\begin{equation*}
\begin{split}
    \nu_{\mathcal{M}}(\mathcal{M} \setminus U) &= \frac{\int_{\mathcal{M} \setminus U} e^{-E/\lambda}\, \diff\mathrm{vol}_{\mathcal{M}}}{Z_{\mathcal{M}}} \\
    &\le \frac{\mathrm{vol}_{\mathcal{M}}(\mathcal{M} \setminus U)}{\mathrm{vol}_{\mathcal{M}}(B)}\, e^{-\delta/\lambda} \xrightarrow[\lambda \to 0^+]{} 0,
\end{split}
\end{equation*}
since $E_U - (E^\star + \delta) = \delta > 0$. This proves (iii).
\end{proof}

\begin{proof}[Proof of Prop.~\ref{prop: regularity_S}]
Since each $f_t$ is smooth, the recursively generated states $x_t(z)$ are finite compositions of smooth maps and hence smooth in $z$; since $h_t, h_T$ are smooth, $G$ in \eqref{eq: G_def} is smooth. By assumption, $\operatorname{rank} J_G(z) = R$ for every $z \in G^{-1}(0)$, so $0$ is a regular value of $G$. The regular value theorem then implies that $\mathcal{X} = G^{-1}(0)$ is a smooth embedded submanifold of $\R^{Tm}$ of codimension $R$, i.e., of dimension $Tm - R$.
\end{proof}

\begin{proof}[Proof of Prop.~\ref{prop: regularity_Mp}]
$(i)$ Given $z = (x_0, u_0, \dots, u_{T-1}) \in \R^{n+Tm}$, the state sequence is generated recursively by $x_{t+1} = f_t(x_t, u_t)$, $t = 0, \dots, T-1$. This defines the rollout map $\Phi: \R^{n+Tm} \to \R^{n_\tau}$,
\begin{equation*}
\Phi(z) = (x_0, u_0, x_1, u_1, \dots, x_T).
\end{equation*}
Since each $f_t$ is smooth, the recursively generated states $x_1(z), \dots, x_T(z)$ are smooth\footnote{Notice that the independent variables $x_0$ and $u_t$ are linear projections of $z$ (e.g., $x_0(z) = P_{x_0} z = [I_{n \times n} \ 0_{n \times Tm}] z$). The intermediate states $x_t(z)$ are thus finite compositions of smooth transition maps. Differentiability is understood w.r.t.\ $z$.}, and therefore $\Phi$ is smooth. Moreover, let $P: \R^{n_\tau} \to \R^{n+Tm}$ be the projection that extracts $(x_0, u_0, \dots, u_{T-1})$ from a full trajectory. Then by definition $P \circ \Phi = \mathrm{Id}_{\R^{n+Tm}}$. Hence the Fr\'echet derivative satisfies
\[
D(P \circ \Phi)(z) = DP(\Phi(z)) D\Phi(z) = I_{n+Tm},
\]
so $D\Phi(z)$ possesses a left inverse, implying $D\Phi(z)$ is injective for every $z$; $\Phi$ is an injective immersion. Also, $P$ restricted to $\Phi(\R^{n+Tm})$ is a continuous inverse of $\Phi$, so $\Phi$ is a homeomorphism onto its image (with the subspace topology) and hence a smooth embedding. Its image is exactly $\mathcal{M}_d$: any $\tau \in \mathcal{M}_d$ satisfies the recursion, so $\tau = \Phi(P(\tau))$, and conversely every rollout satisfies the dynamics. In particular, $\Phi$ is a diffeomorphism from $\R^{n+Tm}$ onto $\mathcal{M}_d$, so $\mathcal{M}_d$ is a smooth embedded submanifold of $\R^{n_\tau}$ of dimension $n + Tm$; the argument imposes no rank condition on the dynamics, whose defect constraints are thus regular everywhere. This proves $(i)$.

$(ii)$ By Prop.~\ref{prop: regularity_S}, $\mathcal{X} = G^{-1}(0)$ is a smooth embedded submanifold of $\R^{Tm}$ of dimension $Tm - R$, so $\{\bar{x}_0\} \times \mathcal{X}$ is a smooth embedded submanifold of $\R^{n+Tm}$ of the same dimension, diffeomorphic to $\mathcal{X}$. For $z = (\bar{x}_0, u)$ with $u \in \R^{Tm}$, the trajectory $\Phi(z)$ satisfies the dynamics and the initial condition by construction, and $u \in \mathcal{X}$ is equivalent to $\Phi(z)$ satisfying all path and terminal equality constraints; therefore $\mathcal{M}_p = \Phi(\{\bar{x}_0\} \times \mathcal{X})$. Since $\Phi$ is an embedding, its restriction to the embedded submanifold $\{\bar{x}_0\} \times \mathcal{X}$ is again an embedding, so $\mathcal{M}_p$ is a smooth embedded submanifold of $\R^{n_\tau}$, diffeomorphic to $\mathcal{X}$, with
\begin{equation*}
\dim \mathcal{M}_p = \dim \mathcal{X} = Tm - R = Tm - \sum_{t=0}^{T-1} r_t - r_T. \qedhere
\end{equation*}
\end{proof}

\begin{corollary}[A stagewise sufficient condition for regularity]
\label{cor: sufficient_regular_terminal}
Let $f_t$, $h_t$ ($t = 0, \dots, T-1$), and $h_T$ be smooth and fix
$x_0 = \bar{x}_0$ as in Sec.~\ref{subsec: rollout_coordinates}. Let
$G_p: \R^{Tm} \to \R^{R_p}$, $R_p := \sum_{t=0}^{T-1} r_t$, collect the
path-constraint blocks of $G$ in \eqref{eq: G_def}, i.e., all rows
except the terminal block $h_T(x_T(z))$.
Assume $r_t \le m$ and $\operatorname{rank} \partial_u h_t(x_t, u_t) = r_t$ for every $t = 0, \dots, T-1$ and every rollout satisfying the path constraints. Let $\mathcal{X}_p := G_p^{-1}(0) \subset \R^{Tm}$ be the reduced path-constrained rollout set, and define the terminal reduced map $K: \mathcal{X}_p \to \R^{r_T}, K(z) = h_T(x_T(z))$.

Assume $\operatorname{rank} \diff K(z)\big|_{T_z \mathcal{X}_p} = r_T$ for all $z \in \mathcal{X}_p$ such that $K(z) = 0$. Then $\operatorname{rank} J_G(z) = R$ for all $z \in \mathcal{X}$, i.e., the rank assumption of Prop.~\ref{prop: regularity_S} holds; if moreover $\mathcal{X} \neq \emptyset$, then $\mathcal{X}$ and $\mathcal{M}_p$ are smooth embedded submanifolds (Props.~\ref{prop: regularity_S} and~\ref{prop: regularity_Mp}$(ii)$) of dimension $\dim \mathcal{X} = \dim \mathcal{M}_p = Tm - R_p - r_T$.
\end{corollary}
\begin{proof}
For each $t = 0, \dots, T-1$, define $G_{p,t}(z) = h_t(x_t(z), u_t)$. Since $x_t(z)$ depends only on $x_0, u_0, \dots, u_{t-1}$, it does not depend on the current control $u_t$. Therefore, $\frac{\partial G_{p,t}}{\partial u_t} = \partial_u h_t(x_t, u_t)$.
Moreover, for $j > t$, $\frac{\partial G_{p,t}}{\partial u_j} = 0$. Hence the Jacobian of $G_p$ with respect to the control variables $(u_0, \dots, u_{T-1})$ has block lower-triangular structure:
\[
J_{G_p}
=
\begin{bmatrix}
\partial_u h_0 & 0 & \cdots & 0\\
* & \partial_u h_1 & \cdots & 0\\
\vdots & \vdots & \ddots & \vdots\\
* & * & \cdots & \partial_u h_{T-1}
\end{bmatrix}.
\]
By assumption, each diagonal block $\partial_u h_t$ has full row rank $r_t$. Selecting $r_t$ linearly independent columns from each diagonal block gives an invertible $R_p \times R_p$ lower-triangular submatrix. Therefore, $\operatorname{rank} J_{G_p}(z) = R_p$ on $G_p^{-1}(0)$. Thus $0$ is a regular value of $G_p$, and the regular value theorem implies that $\mathcal{X}_p = G_p^{-1}(0)$ is a smooth embedded submanifold of $\R^{Tm}$ with codimension $R_p$ and tangent spaces $T_z \mathcal{X}_p = \ker J_{G_p}(z)$.

Now fix $z \in \mathcal{X} = \{z \in \mathcal{X}_p : K(z) = 0\}$ and note that $\diff K(z)$ is the restriction of the ambient Jacobian $J_{h_T \circ x_T}(z)$ to $T_z \mathcal{X}_p$. Suppose there is a linear dependence $v^\T J_{G_p}(z) + w^\T J_{h_T \circ x_T}(z) = 0$ for some $v \in \R^{R_p}$ and $w \in \R^{r_T}$; we must show that $v = 0$ and $w = 0$. Applying this row vector to any $\xi \in T_z \mathcal{X}_p = \ker J_{G_p}(z)$ yields $w^\T \diff K(z)\, \xi = 0$ for all such $\xi$; since $\diff K(z)\vert_{T_z \mathcal{X}_p}$ is surjective by assumption, $w = 0$, and then $v = 0$ because $J_{G_p}(z)$ has full row rank. Hence the stacked Jacobian $J_G(z) = \big[J_{G_p}(z)^\T \ \ J_{h_T \circ x_T}(z)^\T\big]^\T$ has full row rank $R = R_p + r_T$ for every $z \in \mathcal{X}$, which is precisely the rank assumption of Prop.~\ref{prop: regularity_S}. If $\mathcal{X} \neq \emptyset$, Prop.~\ref{prop: regularity_S} then gives that $\mathcal{X}$ is a smooth embedded submanifold of $\R^{Tm}$ of dimension $Tm - R_p - r_T$, and Prop.~\ref{prop: regularity_Mp}$(ii)$ transports the statement to $\mathcal{M}_p = \Phi(\{\bar{x}_0\} \times \mathcal{X})$, of the same dimension. 
\end{proof}

Cor.~\ref{cor: sufficient_regular_terminal} replaces the single coupled rank condition of Prop.~\ref{prop: regularity_S} with $T$ stagewise checks on the small matrices $\partial_u h_t(x_t, u_t) \in \R^{r_t \times m}$. The terminal condition remains genuinely nontrivial, since after the path constraints are satisfied, the surviving feasible directions must still be able to move the terminal constraint in all $r_T$ independent directions.

\section{Algorithm Structures}
\label{app: algorithms}
\begin{algorithm}[H]
\caption{Adjoint Schr\"odinger Bridge Sampling}
\label{alg: asbs}
\begin{algorithmic}[1]
\Require energy $E$, source $\mu$, schedule $\sigma_t$ with $\bar{\sigma}^2 := \int_0^1 \sigma_t^2 \diff t$, time grid $0 = t_0 < \dots < t_N = 1$, epochs per stage $(K_u, K_h)$, buffer size $B$, updates per epoch $L$, step size $\eta$, clipping radius $\gamma$
\Statex \textsc{Rollout}$(u, X_0)$: Euler--Maruyama integration of \eqref{eq: flat_controlled_sde} over the grid, returning $X_1$
\State initialize $\theta$, and $\psi$ such that $h^{\psi} \equiv 0$
\For{stage $= 1, 2, \dots$}
    \Statex \hskip\algorithmicindent \emph{Adjoint matching}: fit $u^{\theta}$ against $h^{\psi}$
    \For{$K_u$ epochs}
        \State $\mathcal{D} \gets \emptyset$
        \While{$\vert \mathcal{D} \vert < B$}
            \State $X_0 \sim \mu$, \ $X_1 \gets$ \textsc{Rollout}$(\bar{u}^{\theta}, X_0)$
            \State $a \gets \mathrm{clip}_{\gamma}\big(\nabla E(X_1)\big) + h^{\psi}(X_1)$ \label{line: adjoint}
            \State $\mathcal{D} \gets \mathcal{D} \cup \{(X_0, X_1, a)\}$
        \EndWhile
        \For{$L$ updates}
            \State draw $(X_0, X_1, a) \sim \mathcal{D}$, \ $t \sim \mathcal{U}[0, 1]$
            \State $X_t \sim p_{t \vert 0, 1}^{\mathrm{base}}(\cdot \vert X_0, X_1)$
            \State $\theta \gets \theta - \eta \nabla_{\theta} \Vert u_t^{\theta}(X_t) + \sigma_t a \Vert^2$ \label{line: am_update}
        \EndFor
    \EndFor
    \Statex \hskip\algorithmicindent \emph{Corrector matching}: fit $h^{\psi}$ against $u^{\theta}$
    \For{$K_h$ epochs}
        \State $\mathcal{D} \gets \emptyset$
        \While{$\vert \mathcal{D} \vert < B$}
            \State $X_0 \sim \mu$, \ $X_1 \gets$ \textsc{Rollout}$(u^{\theta}, X_0)$
            \State $s \gets (X_1 - X_0) / \bar{\sigma}^2$
            \State $\mathcal{D} \gets \mathcal{D} \cup \{(X_1, s)\}$
        \EndWhile
        \For{$L$ updates}
            \State draw $(X_1, s) \sim \mathcal{D}$
            \State $\psi \gets \psi - \eta \nabla_{\psi} \Vert h^{\psi}(X_1) + s \Vert^2$
        \EndFor
    \EndFor
\EndFor
\State \Return $u^{\theta}$ \Comment{$X_1 \sim \nu$ under \eqref{eq: flat_controlled_sde}}
\end{algorithmic}
\end{algorithm}
Algs.~\ref{alg: asbs} and~\ref{alg: rasbs} state the Euclidean sampler of Sec.~\ref{subsec: euclidean_ASBS} and the manifold sampler of Sec.~\ref{subsec: riemannian_ASBS}, each for a generic energy $E$ and source $\mu$. In Alg.~\ref{alg: asbs}, the two regressions \eqref{eq: ASBS_u_star} and \eqref{eq: euclidean_corrector} are run in alternation, each for a fixed number of epochs. Within an epoch, endpoints are simulated once into a buffer and reused across gradient steps, so the number of energy-gradient evaluations is set by the buffer size rather than by the number of parameter updates. 


\begin{algorithm}[H]
\caption{Extended R--ASBS}
\label{alg: rasbs}
\begin{algorithmic}[1]
\Require constraint map $c$ with tangent projector $P_x$ and retraction $\Pi_{\mathcal{M}}$, energy $E$, nominal point $z_{\mathrm{nom}}$ with $c(z_{\mathrm{nom}}) \approx 0$, source covariance $\Sigma$ and width $\kappa$, noise level $\sigma$, grid $t_j = j \Delta t$ with $\Delta t = 1/N$, epochs $K$, batch size $B$, step sizes $(\eta_u, \eta_h)$, number of final samples $M$
\Statex All quantities are batched over $B$ i.i.d.\ copies; losses are batch averages
\State initialize $\theta$, $\psi$
\For{$K$ epochs}
    \State $X_{t_0} \gets \textsc{FeasibleSource}(z_{\mathrm{nom}}, \Sigma, \kappa)$ \label{line: r_source}
    \For{$j = 0, \dots, N-1$} \Comment{Euler--Maruyama}
        \State $\zeta_j \sim \mathcal{N}(0, I)$
        \State $v_j \gets P_{X_{t_j}} \big( \sigma \Delta t\, u^{\theta}_{t_j}(X_{t_j}) + \sigma \sqrt{\Delta t}\, \zeta_j \big)$
        \State $X_{t_{j+1}} \gets \Pi_{\mathcal{M}}(X_{t_j} + v_j)$ \label{line: r_step}
    \EndFor
    \State $a_{t_N} \gets P_{X_{t_N}} \big( \nabla E(X_{t_N}) + h^{\psi}(X_{t_N}) \big)$ \label{line: r_adjoint}
    \For{$j = N-1, \dots, 0$} \Comment{PAT}
        \State $a_{t_j} \gets P_{X_{t_j}} a_{t_{j+1}}$ \label{line: r_pat}
    \EndFor
    \State $\ell_u(\theta) \gets \frac{1}{N} \sum_{j=0}^{N-1} \big\Vert P_{X_{t_j}} u^{\theta}_{t_j}(X_{t_j}) + \sigma a_{t_j} \big\Vert^2$
    \State $\theta \gets \theta - \eta_u \nabla_{\theta} \ell_u(\theta)$ \label{line: r_am}
    \State $X_1' \gets$ lines~\ref{line: r_source}--\ref{line: r_step} from $X_{t_0}$ with updated $\theta$
    \State $s \gets - P_{X_1'} (X_1' - X_{t_0}) / \sigma^2$ \label{line: r_chord}
    \State $\psi \gets \psi - \eta_h \nabla_{\psi} \Vert P_{X_1'} h^{\psi}(X_1') - s \Vert^2$
\EndFor
\State draw $X_1^{(i)}$, $i = 1, \dots, M$, by lines~\ref{line: r_source}--\ref{line: r_step} under $u^{\theta}$
\State \Return $\{X_1^{(i)}\}_{i=1}^{M}$ and $X_1^{\star} = \argmin_{i} \mathcal{J}\big(X_1^{(i)}\big)$
\Statex
\Procedure{FeasibleSource}{$z_{\mathrm{nom}}, \Sigma, \kappa$}
    \State $\varepsilon \sim \mathcal{N}(0, I)$
    \State $\xi \gets z_{\mathrm{nom}} + \kappa\, P_{z_{\mathrm{nom}}} \Sigma^{1/2} \varepsilon$ \Comment{in $T_{z_{\mathrm{nom}}} \mathcal{M}$} \label{line: r_perturb}
    \State \Return $\Pi_{\mathcal{M}}(\xi)$ \Comment{Newton retraction} \label{line: r_retract}
\EndProcedure
\end{algorithmic}
\end{algorithm}

Following the initialization of \cite{liu2023adjoint}, the corrector is initialized at $h^{\psi} \equiv 0$, so that the first adjoint stage regresses on $\nabla E$ alone; the terminal gradient is norm-clipped at a radius $\gamma$, which bounds the regression targets when $\nabla E$ is heavy-tailed without altering the stationary target.

Here $\bar{u} = \texttt{stopgrad}(u)$ as in \eqref{eq: ASBS_u_star}: buffers are populated under the current control without differentiating through the simulation, so no gradient traverses the sampling SDE.



The source $\mu_{\mathcal{M}}$ of Asm.~\ref{ass: sampling_access} is the part of the procedure that cannot be delegated to the sampler, and \textsc{FeasibleSource} makes its construction explicit. It starts from a nominal point $z_{\mathrm{nom}}$ that satisfies the constraints only approximately, perturbs it in the tangent space at $z_{\mathrm{nom}}$ with a Gaussian of covariance $\kappa^2 \Sigma$ (line~\ref{line: r_perturb}), and retracts the result onto $\mathcal{M}$ (line~\ref{line: r_retract}). Projecting the perturbation onto $T_{z_{\mathrm{nom}}} \mathcal{M}$ leaves the retraction only the second-order residual to remove, and $\kappa$ sets the width of the source, i.e., the size of the residual the Newton iterations must absorb. No optimization problem is solved to obtain $z_{\mathrm{nom}}$: in rollout coordinates the dynamics hold by construction, so it suffices to interpolate a reference path through the prescribed waypoints and terminal state, invert the nominal dynamics along it to obtain a control sequence, and, when strict inequalities are encoded by slacks (Sec.~\ref{subsec: assumptions_inequalities}), close the slacks in closed form as $s_t = \log(-g_t)$ evaluated on the nominal rollout. Samples whose retraction does not converge within the iteration budget are not rejected; their constraint residual is monitored instead, and in all instances of Sec.~\ref{sec: numerics} it stays at the retraction tolerance. The same pair $(z_{\mathrm{nom}}, \Sigma)$ also defines the Gaussian term of the energy $E$ in our implementation, so that the source doubles as a prior centered on the nominal trajectory.


\bibliographystyle{IEEEtran}
\bibliography{refs}

\end{document}